\documentclass[11pt,reqno,letterpaper]{amsart}

\usepackage{amssymb}
\usepackage{amsmath}
\usepackage{amsthm}
\usepackage{mathrsfs}
\usepackage{bbm}
\usepackage{doi}
\usepackage{tikz}
\usetikzlibrary{arrows,decorations.pathreplacing}

\usepackage{hyperref}
\hypersetup{pdfstartview={FitH}, hidelinks}
\usepackage{bookmark}

\numberwithin{equation}{section}

\newtheorem{theorem}{Theorem}
\newtheorem{lemma}[theorem]{Lemma}

\renewcommand{\leq}{\leqslant}
\renewcommand{\geq}{\geqslant}

\newcommand{\R}{\mathbb{R}}

\begin{document}

\title{A converse to generalized Runcorn's theorem}

\author[V. Kova\v{c}]{Vjekoslav Kova\v{c}}
\address[V. K.]{University of Zagreb Faculty of Science, Department of Mathematics, Bijeni\v{c}ka cesta 30, 10000 Zagreb, Croatia}
\email{vjekovac@math.hr}

\author[I. Smoli\'{c}]{Ivica Smoli\'{c}}
\address[I. S.]{University of Zagreb Faculty of Science, Department of Physics, Bijeni\v{c}ka cesta 32, 10000 Zagreb, Croatia}
\email{ismolic@phy.hr}



\begin{abstract}
We study an inverse problem for generalized Runcorn's theorem motivated by an apparent paradox of lunar magnetism. In mathematical terms, we characterize square-integrable complex functions $f$ such that
\[ \int_{\{x\in\mathbb{R}^n:a\leq|x|\leq b\}} f(x) \nabla u(x)\cdot\nabla v(x)\,\textup{d}\textup{V}(x) = 0 \]
for every complex harmonic function $u$ in the inner ball $|x|<r_+$ and every complex harmonic function $v$ in the exterior region $|x|>r_-$ that vanishes at infinity. The solution space depends on whether the radii $a$ and $b$ such that $r_-<a<b<r_+$ are regarded as varying or fixed. The solutions are described through the expansion of $f$ into spherical harmonics, and explicit representations are provided.
\end{abstract}

\maketitle


\section{Introduction}

\subsection{Background and motivation}
One of the objectives of NASA's Apollo missions in the 1960s and 1970s was to gain a better understanding of lunar magnetism. As earlier predictions had suggested, various orbiting probes confirmed the lack of a significant global magnetic field, consistent with the absence of an active dynamo in the Moon's core. However, analysis of the lunar rocks brought back by the missions revealed that they possess remanent magnetization. Keith S. Runcorn \cite{Run75a,Run75b} proposed an explanation for this discrepancy, suggesting that the Moon once had an active core dynamo whose magnetic field magnetized the lunar crust. Assuming the crust can be modeled as a spherical shell (true deviations from spherical symmetry are on the order of $10^{-3}$), that the magnetization is linear (with isotropic and possibly radially dependent susceptibility), and that the lunar dynamo is no longer active, the magnetic field in the external region would be exactly zero, despite the presence of a magnetized crust.
Alternative derivations and later geophysical discussions of Runcorn's theorem can be found, for example, in \cite{B86,J07,GIMW11}, while deviations from the perfectly spherical setting are analyzed in \cite{JWL99,LJ00}.

Runcorn's result was generalized by one of the present authors \cite{S25} as follows. Fix an integer $n\geq2$. A function $f$ defined on a domain invariant under rotations is \emph{radial} if $f(x)=g(|x|)$ for some one-dimensional function $g$. Here $|x|$ denotes the Euclidean norm. For arbitrary numbers
\[ 0 < r_- < a < b < r_+ < \infty \]
denote
\begin{subequations}
\begin{align}
\mathscr{U} & := \{x\in\R^n : |x|<r_+\}, \label{eq:defofU} \\
\mathscr{V} & := \{x\in\R^n : |x|>r_-\}, \label{eq:defofV} \\
\mathscr{A}_{a,b} & := \{x\in\R^n : a\leq |x|\leq b\}. \label{eq:defofA}
\end{align}
\end{subequations}
Then for every radial $\textup{C}^1$ function $f$ on $\mathscr{U}\cap\mathscr{V}$, every harmonic function $u$ on $\mathscr{U}$, and every harmonic function $v$ on $\mathscr{V}$ vanishing at infinity (that is,
$v(x)\to0$ as $|x|\to\infty$), one has
\begin{equation}\label{eq:runcorn}
\int_{\mathscr{A}_{a,b}} f(x) \nabla u(x) \cdot \nabla v(x) \,\textup{d}\textup{V}(x) = 0;
\end{equation}
see Figure~\ref{fig:annulus-geometry}.

\begin{figure}[ht]
\centering
\begin{tikzpicture}[>=stealth', scale=1.1]
 \def\rminus{0.8}
 \def\ra{1.45}
 \def\rb{2.15}
 \def\rplus{2.85}
 \fill[gray!18, even odd rule] (0,0) circle (\rb) (0,0) circle (\ra);
 \draw[dashed] (0,0) circle (\rminus);
 \draw[thick] (0,0) circle (\ra);
 \draw[thick] (0,0) circle (\rb);
 \draw[dashed] (0,0) circle (\rplus);
 \draw[->] (0,0) -- (3.7,0) node[below] {$r$};
 \draw[fill=white] (0,0) circle (0.03);
 \node at (0.22+\rminus,0.1) {$r_{\!{\scriptscriptstyle -}}$};
 \node at (0.14+\ra,0.13) {$a$};
 \node at (0.14+\rb,0.16) {$b$};
 \node at (0.22+\rplus,0.13) {$r_{\!{\scriptscriptstyle +}}$};
 \node at (0,2.45) {$\mathscr{A}_{a,b}$};
 \draw[decorate, decoration={brace, mirror, amplitude=8pt}] (-\rplus,-2.5) -- (0,-2.5);
 \node[align=center] at (-1.5,-3.2) {$u$ harmonic\\in $|x|<r_+$};
 \draw[decorate, decoration={brace, mirror, amplitude=8pt}] (\rminus,-2.5) -- (4,-2.5);
 \draw[white, fill=white] (3.8,-2.7) rectangle (4.2,-2.3);
 \node[align=center] at (2.35,-3.2) {$v$ harmonic\\in $|x|>r_-$};
\end{tikzpicture}
\caption{Cross-section of the annulus.}
\label{fig:annulus-geometry}
\end{figure}
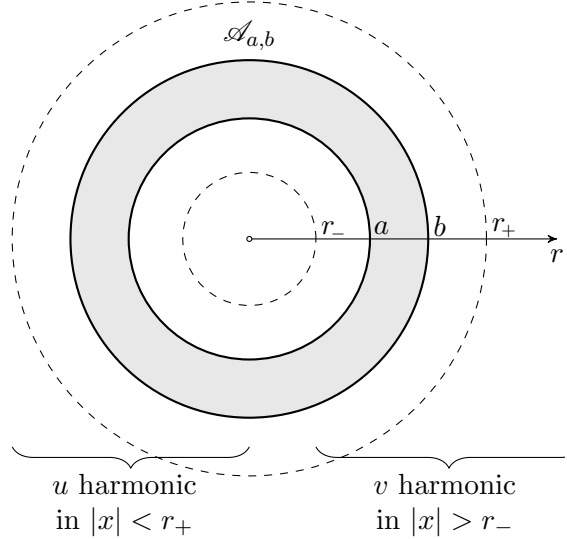

In \cite{S25} the functions $f$, $u$, and $v$ are real-valued, but the same identity remains valid in the complex-valued setting once we interpret
\[ \nabla u\cdot \nabla v := \sum_{m=1}^n \partial_m u\,\partial_m v \]
bilinearly, that is, without complex conjugation.
If the gradients are identified with column vectors and $(\cdot)^\textup{T}$ stands for the transposition, then the above expression can also be written simply as $(\nabla u)^\textup{T} \nabla v$.
Throughout the paper, harmonic functions may be complex-valued; equivalently, one may test the identities on real-valued harmonic functions and then extend them by complex bilinearity.
Moreover, by approximating an arbitrary radial $f\in\textup{L}^1(\mathscr{A}_{a,b})$ with smooth radial functions, one extends \eqref{eq:runcorn} from $\textup{C}^1$ data to all integrable radial $f$.
In what follows it is more convenient to work either with $\textup{C}(\mathscr{A}_{a,b})$, the space of complex continuous functions on $\mathscr{A}_{a,b}$, or with $\textup{L}^2(\mathscr{A}_{a,b})$, the space of complex square-integrable functions on the same annulus.

It is natural to ask for a converse: which square-integrable functions $f$ satisfy \eqref{eq:runcorn} for all admissible harmonic pairs $u,v$?
In three dimensions Arkani-Hamed and Dyment \cite{AHD96} suggested that such an $f$ should necessarily be radial, even when $v$ ranges only over a more restricted class of exterior harmonic functions.
The discussion in \cite{AHD96} is informal and does not settle the question.

Our first main result, Theorem~\ref{thm:varying} below, shows that the expected characterization does hold if \eqref{eq:runcorn} is required on every sub-annulus $\mathscr{A}_{a,b}$ with radii $a$ and $b$ in a prescribed interval $(r_-,r_+)$.
The second main result, Theorem~\ref{thm:fixed}, gives a complete description on a fixed annulus $\mathscr{A}_{a,b}$ and shows that the converse then fails strongly: after decomposition into spherical harmonics, each tangential term is constrained only by finitely many radial moment conditions.
In particular, the solution space is an infinite orthogonal sum of infinite-dimensional subspaces, which we describe explicitly below.
This complements the broader non-uniqueness phenomena that appear in magnetization inversion; compare, for instance, \cite{MH03,GIMW11}.
Further mathematically oriented work on uniqueness, non-uniqueness, and decomposition in inverse magnetization problems includes thin-plate results \cite{BHLSW13}, uniqueness for induced spherical magnetization \cite{Ger16}, source separation problems for geomagnetic potentials \cite{BG17}, decompositions on Lipschitz surfaces \cite{BGK21}, uniqueness on special volumetric models \cite{BGKM21}, and the $\textup{L}^p$ theory of silent sources \cite{BLN25}. Related Hardy-space and incomplete data formulations on the sphere and shell appear in \cite{ABLP10}, while recent work on internal--external field separation under partial spherical data is given in \cite{HGR26}.

\subsection{Statements of the main results}
In all that follows we fix an ambient dimension, which is an integer $n\geq 2$, and two positive reals $r_-<r_+$.
Recall the definitions \eqref{eq:defofU} and \eqref{eq:defofV} of the open sets $\mathscr{U}$ and $\mathscr{V}$, and the definition \eqref{eq:defofA} of the annulus $\mathscr{A}_{a,b}$ with inner radius $a$ and outer radius $b$.

The characterization for varying $a$ and $b$ is comparatively simple.

\begin{theorem}\label{thm:varying}
The following are equivalent for a complex-valued function $f\in\textup{C}(\mathscr{U}\cap\mathscr{V})$.
\begin{enumerate}
\item \label{thmit1} 
Equality \eqref{eq:runcorn} holds for every harmonic $u$ in $\mathscr{U}$, every harmonic $v$ in $\mathscr{V}$ that vanishes at infinity, and every pair of real numbers $a<b$ from the interval $(r_-,r_+)$.
\item \label{thmit2} 
Equality
\[ \int_{\{x\in\R^n : |x|=r\}} f(x) \, \nabla u(x) \cdot \nabla v(x) \,\textup{d}\textup{S}(x) = 0 \]
holds for every harmonic $u$ in $\mathscr{U}$, every harmonic $v$ in $\mathscr{V}$ that vanishes at infinity, and every $r\in(r_-,r_+)$.
\item \label{thmit3}
If $n\geq3$, then $f$ is radial. If $n=2$, then $f$ is of the form
\[ f(r\cos\varphi,r\sin\varphi) = g_{-1}(r) e^{-\mathbbm{i}\varphi} + g_{0}(r) + g_{1}(r) e^{\mathbbm{i}\varphi} \]
for some continuous functions $g_{-1},g_0,g_1\colon(r_-,r_+)\to\mathbb{C}$.
\end{enumerate}
\end{theorem}

We now turn to the fixed-annulus problem, which is more difficult.
Let $\mathscr{H}_k$ denote the space of \emph{spherical harmonics} on $\mathbb{S}^{n-1}\subset\R^n$ of degree $k\geq0$, i.e., the restrictions to $\mathbb{S}^{n-1}$ of harmonic homogeneous polynomials of degree $k$ on $\R^n$.
It is well known \cite[Prop.~5.8]{ABR01} that
\[ d(k) := \dim\mathscr{H}_k = \begin{cases} 1 & \text{if } k=0, \\ n & \text{if } k=1, \\ \binom{n+k-1}{n-1}-\binom{n+k-3}{n-1} & \text{if } k\geq2. \end{cases} \]
Let
\[ Y_{k,1}, Y_{k,2}, \ldots, Y_{k,d(k)} \]
be an orthogonal basis of $\mathscr{H}_k$ with respect to the inner product on $\textup{L}^2(\mathbb{S}^{n-1})$,
\[ \langle G,H\rangle_{\textup{L}^2(\mathbb{S}^{n-1})} := \int_{\mathbb{S}^{n-1}} G(\sigma) \overline{H(\sigma)} \,\textup{d}\textup{S}(\sigma). \]
It is often convenient to allow $Y_{k,j}$ to be complex-valued. However, if one instead wants to give a characterization for real functions $f$, then one should choose $Y_{k,j}$ to be real-valued.
In the radial direction, we will find it convenient to use the weighted $\textup{L}^2$-space $\textup{L}^2\big([a,b],r^{n-1}\,\textup{d}r\big)$, consisting of all equivalence classes modulo equality almost everywhere of measurable functions $g\colon[a,b]\to\mathbb{C}$ satisfying
\[ \int_a^b |g(r)|^2 \,r^{n-1}\,\textup{d}r < \infty. \]
It is equipped with the inner product
\[ \langle g,h\rangle_{\textup{L}^2([a,b],r^{n-1}\textup{d}r)} := \int_a^b g(r) \overline{h(r)} \,r^{n-1}\,\textup{d}r. \]

\begin{theorem}\label{thm:fixed}
For $k\geq0$ define a finite set $\mathfrak{R}_k$ of monomials in variable $r$ by
\[ \mathfrak{R}_k :=
\begin{cases}
\emptyset & \text{if } n\geq2,\ k=0 \text{ or } n=2,\ k=1,\\[1mm]
\{r^{2\ell-k} : 0\leq\ell\leq k-2\} & \text{if } n=2,\ k\geq2,\\[1mm]
\{r^{2\ell-k-n+2} : 0\leq\ell\leq k-1\} & \text{if } n\geq3,\ k\geq1.
\end{cases} \]
Fix real numbers $a,b$ such that $r_-<a<b<r_+$, and set
\[ \mathfrak{G}_k := \Big\{ g\in \textup{L}^2\big([a,b],r^{n-1}\,\textup{d}r\big) : \underbrace{\int_a^b g(r) \overline{h(r)} \, r^{n-1} \,\textup{d}r}_{\langle g,h\rangle_{\textup{L}^2([a,b],r^{n-1}\textup{d}r)}} = 0 \text{ for every } h\in\mathfrak{R}_k \Big\}. \]
Then a complex-valued function $f\in\textup{L}^2(\mathscr{A}_{a,b})$ satisfies \eqref{eq:runcorn} for every harmonic $u$ on $\mathscr{U}$ and every harmonic $v$ on $\mathscr{V}$ vanishing at infinity if and only if it can be expanded in $n$-dimensional spherical coordinates as
\begin{equation}\label{eq:sph_exp}
f(r\sigma) = \sum_{k=0}^{\infty} \sum_{j=1}^{d(k)} g_{k,j}(r) Y_{k,j}(\sigma); \quad r\in[a,b], \ \sigma\in\mathbb{S}^{n-1},
\end{equation}
with convergence in $\textup{L}^2(\mathscr{A}_{a,b})$, where $g_{k,j}\in\mathfrak{G}_k$ for every $k\geq0$ and every $1\leq j\leq d(k)$.
\end{theorem}

The elements of $\mathfrak{R}_k$ are Laurent monomials, i.e., integer powers of $r$.

Observe that, for almost every fixed $r\in[a,b]$, equality \eqref{eq:sph_exp} constitutes an expansion of $f(r\cdot)$ into spherical harmonics, so their mutual orthogonality on $\mathbb{S}^{n-1}$ implies 
\begin{equation}\label{eq:sph_exp_coeff}
g_{k,j}(r)=\frac{1}{\|Y_{k,j}\|_{\textup{L}^2(\mathbb{S}^{n-1})}^2}
\int_{\mathbb{S}^{n-1}} f(r\sigma)\overline{Y_{k,j}(\sigma)}\,\textup{d}\textup{S}(\sigma)
\end{equation}
for a.e.\@ $r\in[a,b]$.
Therefore, an equivalent characterization of the square-integrable functions $f$ satisfying \eqref{eq:runcorn} is obtained by requiring that \eqref{eq:sph_exp_coeff} belongs to $\mathfrak{G}_k$ for every $k\geq0$ and every $1\leq j\leq d(k)$.

In view of Theorem~\ref{thm:fixed}, yet another way of rewriting the solution space of \eqref{eq:runcorn} on the fixed annulus $\mathscr{A}_{a,b}$ is as the Hilbert-space orthogonal sum
\[ \bigoplus_{k=0}^{\infty} \,\bigl(\mathfrak{G}_k\otimes \mathscr{H}_k\bigr), \]
where $\otimes$ stands for the radial--tangential tensor splitting of $n$-dimensional functions.
Since the elements of $\mathfrak{R}_k$ have pairwise distinct exponents, they are linearly independent functions on every nondegenerate interval $[a,b]$. Hence, $\mathfrak{G}_k$, as the orthogonal complement of the linear span of $\mathfrak{R}_k$ with respect to the inner product $\langle\cdot,\cdot\rangle_{\textup{L}^2([a,b],r^{n-1}\textup{d}r)}$, satisfies
\[ \operatorname{codim}\mathfrak{G}_k = \dim\operatorname{span}\mathfrak{R}_k = \mathop{\textup{card}}\mathfrak{R}_k. \]
In particular, $\mathfrak{G}_0=\textup{L}^2\big([a,b],r^{n-1}\,\textup{d}r\big)$, so the degree-zero subspace is without any constraints, recovering generalized Runcorn's theorem, while every higher-order tangential term still contributes an infinite-dimensional family of solutions.

The proof of Theorem~\ref{thm:fixed} relies on several classical ingredients from harmonic analysis, such as the Kelvin transform of a harmonic function and the harmonic decomposition of a general homogeneous polynomial; see \eqref{eq:SWdecompose} below. However, it takes a nontrivial effort to characterize the exact linear span of the inner products $\nabla u\cdot\nabla v$, and for this purpose we use a less standard spanning set of harmonic homogeneous polynomials, parametrized by complex $n$-tuples $(\zeta_1,\ldots,\zeta_n)$ such that $\sum_m \zeta_m^2 = 0$; see Lemma~\ref{lem:span}.
Theorem~\ref{thm:varying} is then recovered simply by applying the fixed-annulus classification to every sub-annulus of $\mathscr{U}\cap\mathscr{V}$.

For later use we record some low-dimensional cases explicitly.
\begin{itemize}
\item If $n=2$, then $\mathfrak{R}_0=\mathfrak{R}_1=\emptyset$ and
\[ \mathfrak{R}_k = \{r^{2\ell-k}:0\leq \ell\leq k-2\}; \quad k\geq2. \]
Hence $\operatorname{codim}\mathfrak{G}_k = k-1$ for $k\geq2$.
\item If $n=3$, then $\mathfrak{R}_0=\emptyset$ and
\[ \mathfrak{R}_k = \{r^{2\ell-k-1}:0\leq \ell\leq k-1\}; \quad k\geq1. \]
Hence $\operatorname{codim}\mathfrak{G}_k = k$ for $k\geq1$.
\end{itemize}
These are discussed in detail in Section~\ref{sec:examples}.

As we have already mentioned, several papers on magnetization and/or inverse problems are related to our work through uniqueness, non-uniqueness, and source decompositions. None of them directly proves Theorems~\ref{thm:varying} or \ref{thm:fixed}, but some of the literature develops related theories in much greater depth. Our aim is instead to give a concrete mathematical discussion of a particular problem coming from physics.


\section{Several auxiliary results}

This section contains three auxiliary lemmas. The first one identifies a convenient spanning family of harmonic polynomials, the second one computes the gradient product $\nabla u \cdot \nabla v$ for special pairs $u,v$, while the third (and the most substantial) one computes the linear span of more general dot products $\nabla u \cdot \nabla v$.
Each of them is a consequence of the existing literature on harmonic functions \cite{SW71,ABR01} and some computation. Still, we include their full proofs, because we need to prepare them for their application in the next section.

For $k\geq0$ let $\mathscr{P}_k$ denote the vector space of complex homogeneous polynomials of degree $k$ on $\R^n$. Let $\mathcal{H}_k\subseteq\mathscr{P}_k$ denote the subspace of harmonic ones; these are called \emph{solid spherical harmonics} \cite[\S IV.2]{SW71}.
Restricting a polynomial from $\mathcal{H}_k$ to $\mathbb{S}^{n-1}$ identifies $\mathcal{H}_k$ with $\mathscr{H}_k$.

The complex vector space $\mathscr{P}_k$ is equipped with the inner product $\langle\cdot,\cdot\rangle_{\mathscr{P}_k}$ defined as
\[ \Bigl\langle \sum_{|\alpha|=k} c_\alpha x^\alpha, \sum_{|\alpha|=k} d_\alpha x^\alpha \Bigr\rangle_{\mathscr{P}_k}
:= \sum_{|\alpha|=k} \alpha! \,c_\alpha \overline{d_\alpha}. \]
Here $\alpha$ is an $n$-tuple $(\alpha_1,\alpha_2,\ldots,\alpha_n)$ of non-negative integers,
\[ |\alpha| := \alpha_1 + \alpha_2 + \cdots + \alpha_n, \]
\[ \alpha! := \alpha_1! \alpha_2! \cdots \alpha_n! \]
(with $0!:=1$), and for $x=(x_1,x_2,\ldots,x_n)$ we set
\[ x^\alpha := x_1^{\alpha_1} x_2^{\alpha_2} \cdots x_n^{\alpha_n}. \]
If we also agree to use the partial differential notation
\[ D^\alpha := \partial_1^{\alpha_1} \partial_2^{\alpha_2} \cdots \partial_n^{\alpha_n}
= \frac{\partial^{\alpha_1 + \alpha_2 + \cdots + \alpha_n}}{\partial x_1^{\alpha_1} \partial x_2^{\alpha_2} \cdots \partial x_n^{\alpha_n}}, \]
then the same inner product on $\mathscr{P}_k$ can equivalently be written as
\[ \langle P,Q\rangle_{\mathscr{P}_k} = P(D) \overline{Q}; \]
see \cite[\S IV.2]{SW71}.

\begin{lemma}\label{lem:span}
For every integer $k\geq0$, the space $\mathcal{H}_k$ is spanned by the polynomials
\begin{equation}\label{eq:zeta_poly}
x\mapsto (\zeta\cdot x)^k,
\end{equation}
where $\zeta\in\mathbb{C}^n$ satisfies $\zeta\cdot\zeta=0$.
\end{lemma}

We recall that the condition $\zeta\cdot\zeta=0$ should be read as
$\sum_m \zeta_m^2 = 0$ if $\zeta$ is given by its coordinates $(\zeta_1,\ldots,\zeta_n)$, or, equivalently, as $\zeta^\textup{T} \zeta=0$ if $\zeta$ is understood to be a column vector.

\begin{proof}
This exact result can be found in the literature, e.g., as \cite[Lemma~4.10]{Flavi}. However, since its formulation is not entirely standard, we give a self-contained elementary proof, different from the one in \cite{Flavi}.

The claim is trivial for $k\leq1$, so assume that $k\geq2$.
Let $\mathcal{E}_k$ be the span of the polynomials \eqref{eq:zeta_poly} as $\zeta\in\mathbb{C}^n$ ranges over complex tuples such that $\zeta\cdot\zeta=0$.
All these polynomials are harmonic because of
\begin{equation}\label{eq:form_harm}
\Delta (\zeta\cdot x)^k = k(k-1)(\zeta\cdot\zeta)(\zeta\cdot x)^{k-2}
\end{equation}
and $\zeta\cdot\zeta=0$, so $\mathcal{E}_k\subseteq\mathcal{H}_k$.

In order to prove the reverse set inclusion, we take a polynomial $P\in\mathcal{H}_k$ that is orthogonal to $\mathcal{E}_k$ with respect to the inner product $\langle\cdot,\cdot\rangle_{\mathscr{P}_k}$. 
Write it explicitly as
\[ P(x) = \sum_{|\alpha|=k} c_\alpha x^\alpha \]
for complex coefficients $c_\alpha$.
By the multinomial theorem,
\begin{align*}
\langle P(x), (\zeta\cdot x)^k\rangle_{\mathscr{P}_k} 
& = \Bigl\langle \sum_{\substack{\alpha=(\alpha_1,\ldots,\alpha_n)\\ |\alpha|=k}} c_\alpha x_1^{\alpha_1} \cdots x_n^{\alpha_n}, \sum_{\substack{\alpha=(\alpha_1,\ldots,\alpha_n)\\ |\alpha|=k}} \frac{k!}{\alpha!} (\zeta_1 x_1)^{\alpha_1} \cdots (\zeta_n x_n)^{\alpha_n} \Bigr\rangle_{\mathscr{P}_k} \\
& = k! \sum_{\substack{\alpha=(\alpha_1,\ldots,\alpha_n)\\ |\alpha|=k}} c_\alpha (\overline{\zeta_1})^{\alpha_1} \cdots (\overline{\zeta_n})^{\alpha_n} = k! P(\overline{\zeta}).
\end{align*}
Regard $P$ as an element of the polynomial ring $\mathbb{C}[z_1,\ldots,z_n]$. Since complex
conjugation is a bijection of the complex quadric $\sum_{m=1}^n z_m^2=0$ onto itself, the preceding identity shows that $P(z)=0$ holds whenever $z\cdot z=0$. We now argue that this implies divisibility of $P(z)$ by the quadratic polynomial $z\cdot z=\sum_{m=1}^n z_m^2$.

Set $z'=(z_1,\ldots,z_{n-1})$ and
\[ S(z'):=\sum_{m=1}^{n-1}z_m^2. \]
Since $z_n^2+S(z')$ is monic as a polynomial in $z_n$, division in the ring $\mathbb{C}[z_1,\ldots,z_{n-1}][z_n]$ gives
\[ P(z) = \bigl(z_n^2+S(z')\bigr)Q(z) + A(z')z_n + B(z') \]
for some $Q\in\mathbb{C}[z_1,\ldots,z_n]$ and $A,B\in\mathbb{C}[z_1,\ldots,z_{n-1}]$.
Fix $z'$ such that $S(z')\neq0$, and choose $\tau\in\mathbb{C}\setminus\{0\}$ satisfying $\tau^2=-S(z')$. Since both $(z',\tau)$ and $(z',-\tau)$ belong to the quadric $z\cdot z=0$, we have
\[ A(z')\tau+B(z')=0, \quad -A(z')\tau+B(z')=0. \]
Thus $A(z')=B(z')=0$. Consequently, $A$ and $B$ vanish on the nonempty open set $\{z':S(z')\neq0\}$, and hence vanish identically.
Therefore, 
\begin{equation}\label{eq:PzQ}
P(z)=(z\cdot z)Q(z)
\end{equation}
and $Q$ is a homogeneous polynomial of degree $k-2$.

By restricting $z=(z_1,\ldots,z_n)$ in \eqref{eq:PzQ} to be a real $n$-tuple, $z=x\in\R^n$, we can now write
\[ P(x) = |x|^2 Q(x) \]
for some $Q\in\mathscr{P}_{k-2}$.
However,
\[ \langle P, P \rangle_{\mathscr{P}_k} = \langle |x|^2 Q(x), P(x) \rangle_{\mathscr{P}_k} 
= Q(D) \Delta \overline{P} = \langle Q, \underbrace{\Delta P}_{=0} \rangle_{\mathscr{P}_{k-2}} = 0 \]
implies that $P$ is the nul-polynomial. We conclude that $\mathcal{E}_k=\mathcal{H}_k$.
\end{proof}

The \emph{Kelvin transform} of an $n$-dimensional function $v$ with respect to the unit sphere $\mathbb{S}^{n-1}$ is defined as
\[ (\mathcal{K}v)(x) := \frac{1}{|x|^{n-2}} v\Bigl(\frac{x}{|x|^2}\Bigr). \]
It is a standard fact in potential theory that $v$ is harmonic on an open set $\Omega\subseteq\R^n\setminus\{0\}$ if and only if $\mathcal{K}v$ is harmonic on $\{x : x/|x|^2\in\Omega\}$; see \cite[Thm.~4.7]{ABR01}.
If $Q\in\mathcal{H}_k$, then its Kelvin transform takes a simpler form,
\begin{equation}\label{eq:Kelvin_of_Q}
(\mathcal{K}Q)(x) = |x|^{2-n-2k} Q(x); \quad x\in\R^n\setminus\{0\}.
\end{equation}
Now we also know that $\mathcal{K}Q$ is harmonic on $\R^n\setminus\{0\}$. If $n=2$, then $\mathcal{K}Q(x)\to0$ as $|x|\to\infty$ whenever $k\geq1$. If $n\geq3$, this happens already for $k\geq0$; see \cite[Ch.~4]{ABR01}.

\begin{lemma}\label{lem:specialpair}
Let $\zeta\in\mathbb{C}^n$ satisfy $\zeta\cdot\zeta=0$. Take an integer $p\geq1$. If $n=2$ take an integer $q\geq1$; if $n\geq3$ take an integer $q\geq0$. Define
\begin{align*} 
u(x) & := (\zeta\cdot x)^p; \quad x\in\R^n, \\ 
v(x) & := |x|^{2-n-2q}(\zeta\cdot x)^q; \quad x\in\R^n\setminus\{0\}. 
\end{align*}
Then $u$ is harmonic on $\R^n$, $v$ is harmonic on $\R^n\setminus\{0\}$, $v$ vanishes at infinity, and
\begin{equation}\label{eq:specialpair}
 \nabla u(x)\cdot\nabla v(x)
 = \underbrace{-p\,(n+2q-2)}_{\neq0}\,|x|^{-n-2q}(\zeta\cdot x)^{p+q}; \quad x\in\R^n\setminus\{0\}.
\end{equation}
\end{lemma}

\begin{proof}
Formula \eqref{eq:form_harm} and $\zeta\cdot\zeta=0$ show that $u$ is harmonic.
Also, $v$ is the Kelvin transform of the harmonic homogeneous polynomial $(\zeta\cdot x)^q$, and, hence, it is harmonic on $\R^n\setminus\{0\}$.
The decay at infinity is immediate from the homogeneity of degree $2-n-q$, together with the restriction $q\geq1$ when $n=2$.
Finally,
\[ \nabla u(x)=p(\zeta\cdot x)^{p-1}\zeta \]
and
\[ \nabla v(x)=(2-n-2q)|x|^{-n-2q}(\zeta\cdot x)^q x + q|x|^{2-n-2q}(\zeta\cdot x)^{q-1}\zeta. \]
Taking the dot product and using $\zeta\cdot\zeta=0$ gives exactly \eqref{eq:specialpair}.
\end{proof}

A standard decomposition (see \cite[Thm.~2.1]{SW71} or \cite[Thm.~5.7]{ABR01}) writes an arbitrary $P\in\mathscr{P}_k$ uniquely as the sum
\begin{equation}\label{eq:SWdecompose}
P(x) = \sum_{\ell=0}^{\lfloor k/2\rfloor} |x|^{2\ell} H_{k-2\ell}(x); \quad x\in\R^n,
\end{equation}
for some spherical harmonics $H_{k-2\ell}\in\mathcal{H}_{k-2\ell}$ indexed by $0\leq\ell\leq \lfloor k/2\rfloor$.
Moreover, \cite[Thm.~5.21]{ABR01} expresses $H_{k-2\ell}(x)$ as a linear combination of
\begin{equation}\label{eq:SWdecompose2}
|x|^{2(j-\ell)} \Delta^j P(x); \quad \ell\leq j\leq \lfloor k/2\rfloor.
\end{equation}

\begin{lemma}\label{lem:decompose}
\begin{itemize}
\item[(a)] Let $n\geq3$, $p\geq1$, $q\geq0$ be integers and let $P\in\mathcal{H}_p$, $Q\in\mathcal{H}_q$ be solid spherical harmonics. Then, for $x\in\R^n\setminus\{0\}$,
\[ \nabla P(x)\cdot\nabla (\mathcal{K}Q)(x)
= \sum_{\ell=0}^{\min\{p-1,q\}} |x|^{-n-2q+2\ell} H_{p+q-2\ell}(x), \]
where $H_{p+q-2\ell}\in\mathcal{H}_{p+q-2\ell}$ for each index $\ell$.
\item[(b)] Let $n=2$, $p\geq1$, $q\geq1$ be integers and let $P\in\mathcal{H}_p$, $Q\in\mathcal{H}_q$ be solid spherical harmonics. Then, for $x\in\R^2\setminus\{0\}$,
\[ \nabla P(x)\cdot\nabla (\mathcal{K}Q)(x)
= \sum_{\ell=0}^{\min\{p-1,q-1\}} |x|^{-2-2q+2\ell} H_{p+q-2\ell}(x), \]
where $H_{p+q-2\ell}\in\mathcal{H}_{p+q-2\ell}$ for each index $\ell$.
\end{itemize}
\end{lemma}

\begin{proof}
We simultaneously prove both parts of the lemma.
Let us first record the following elementary identity. If $H\in\mathcal{H}_k$ and $\lambda\in\mathbb{R}$, then
\begin{equation}\label{eq:radial_harmonic}
\Delta\bigl(|x|^\lambda H(x)\bigr)
= \lambda\bigl(\lambda+2k+n-2\bigr)|x|^{\lambda-2}H(x); \quad x\in\mathbb{R}^n\setminus\{0\}.
\end{equation}
Indeed,
\[ \nabla |x|^\lambda=\lambda |x|^{\lambda-2}x, \quad \Delta |x|^\lambda=\lambda(\lambda+n-2)|x|^{\lambda-2}, \]
and
\[ x\cdot\nabla H(x) = \frac{\textup{d}}{\textup{d}t}\bigg|_{t=1} H(tx) = \frac{\textup{d}}{\textup{d}t}\bigg|_{t=1} t^k H(x) = k H(x). \]
We also recall the harmonicity of $H$, i.e., $\Delta H=0$.
The product rule therefore yields
\begin{align*}
\Delta\bigl(|x|^\lambda H(x)\bigr)
& =(\Delta |x|^\lambda)H(x) +2\nabla |x|^\lambda\cdot\nabla H(x) + |x|^\lambda\Delta H(x) \\
& =\lambda(\lambda+n-2)|x|^{\lambda-2}H(x) + 2\lambda |x|^{\lambda-2}x\cdot\nabla H(x) \\
& =\lambda(\lambda+2k+n-2)|x|^{\lambda-2}H(x),
\end{align*}
which proves \eqref{eq:radial_harmonic}.

Denote
\[ d:=p+q, \quad m:=\min\{p,q\}. \]
By the standard harmonic decomposition \eqref{eq:SWdecompose}, there are unique polynomials
\[ K_{d-2\ell}\in\mathcal{H}_{d-2\ell}; \quad 0\leq\ell\leq\lfloor d/2\rfloor \]
such that
\begin{equation}\label{eq:PQ_full}
P(x)Q(x) = \sum_{\ell=0}^{\lfloor d/2\rfloor} |x|^{2\ell}K_{d-2\ell}(x).
\end{equation}
However, the optimal summation limits are more subtle and we claim that
\begin{equation}\label{eq:PQ_high}
K_{d-2\ell}=0 \quad\text{whenever }\ell>m.
\end{equation}
Namely, for every nonnegative integer $s$ define
\[ R_s := \sum_{i_1,\ldots,i_s=1}^n
\bigl(\partial_{i_1}\cdots\partial_{i_s}P\bigr)
\bigl(\partial_{i_1}\cdots\partial_{i_s}Q\bigr), \]
with the convention $R_0=PQ$. Every partial derivative of $P$ and of $Q$ is harmonic. Hence the product rule gives $\Delta R_s=2R_{s+1}$, so we inductively get $\Delta^s(PQ)=2^sR_s$. Since $m+1$ is larger than at least one of the degrees $p$ and $q$, every summand in $R_{m+1}$ vanishes. Thus
\[ \Delta^{m+1}(PQ)=0. \]
Combining this with \eqref{eq:SWdecompose} and \eqref{eq:SWdecompose2} proves \eqref{eq:PQ_high}. 
Consequently, we can truncate decomposition \eqref{eq:PQ_full} to
\begin{equation}\label{eq:PQ_truncated}
P(x)Q(x) = \sum_{\ell=0}^{m}|x|^{2\ell}K_{d-2\ell}(x).
\end{equation}

Since $Q\in\mathcal{H}_q$, formula \eqref{eq:Kelvin_of_Q} applies to it with $k=q$.
Both $P$ and $\mathcal{K}Q$ are harmonic on
$\mathbb{R}^n\setminus\{0\}$, so
\[ 2\nabla P\cdot\nabla(\mathcal{K}Q) = \Delta\bigl(P\,\mathcal{K}Q\bigr). \]
Using \eqref{eq:Kelvin_of_Q} and \eqref{eq:PQ_truncated}, we obtain
\[ 2\nabla P(x)\cdot\nabla(\mathcal{K}Q)(x)
= \Delta\Bigl(|x|^{2-n-2q}P(x)Q(x)\Bigr) 
= \sum_{\ell=0}^{m} \Delta\Bigl( |x|^{2-n-2q+2\ell}K_{d-2\ell}(x) \Bigr). \]
Applying \eqref{eq:radial_harmonic} with $\lambda=2-n-2q+2\ell$ and $k=d-2\ell$ gives
\[ \Delta\Bigl( |x|^{2-n-2q+2\ell}K_{d-2\ell}(x) \Bigr) 
= -2(p-\ell)(n-2+2q-2\ell) |x|^{-n-2q+2\ell}K_{d-2\ell}(x). \]
After division by $2$, we arrive at
\begin{equation}\label{eq:graddec}
\nabla P(x)\cdot\nabla(\mathcal{K}Q)(x)
= -\sum_{\ell=0}^{m} (p-\ell)(n-2+2q-2\ell) |x|^{-n-2q+2\ell}K_{d-2\ell}(x).
\end{equation}
Both cases, $n\geq3$ and $n=2$, now follow by observing that, additionally, the coefficient in \eqref{eq:graddec} is zero when either $\ell=p$ or $n=2$ and $\ell=q$.
\end{proof}


\section{Proofs of the main theorems}

\begin{proof}[Proof of Theorem~\ref{thm:fixed}]
By the standard orthogonal decomposition
\[ \textup{L}^2(\mathscr{A}_{a,b}) = \bigoplus_{k=0}^{\infty} \big(\textup{L}^2\big([a,b],r^{n-1}\,\textup{d}r\big)\otimes \mathscr{H}_k\big), \]
see \cite[Thm.~5.12]{ABR01}, every $f\in \textup{L}^2(\mathscr{A}_{a,b})$ has a unique expansion of the form
\[ f(r\sigma) = \sum_{k=0}^{\infty} \sum_{j=1}^{d(k)} g_{k,j}(r)Y_{k,j}(\sigma), \]
converging in $\textup{L}^2(\mathscr{A}_{a,b})$, where $g_{k,j}(r)$ is given by formula \eqref{eq:sph_exp_coeff} for a.e.\@ $r\in[a,b]$.
We prove the stated characterization in two steps.

\smallskip
\emph{Necessity.}
Assume that $f$ satisfies \eqref{eq:runcorn} for every admissible pair $u,v$.
Fix $k\geq0$ and $1\leq j\leq d(k)$.
If $\mathfrak{R}_k=\emptyset$, there is nothing to prove for this pair $k,j$.
Otherwise choose an integer $q$ so that $r^{k-n-2q}\in\mathfrak{R}_k$.
Thus, by the definition of $\mathfrak{R}_k$,
\begin{align*}
n=2 & \implies 1\leq q\leq k-1,\\
n\geq3 & \implies 0\leq q\leq k-1.
\end{align*}
For every $\zeta\in\mathbb{C}^n$ such that $\zeta\cdot\zeta=0$, Lemma~\ref{lem:specialpair} with $p=k-q$ gives
\[ \int_{\mathscr{A}_{a,b}} f(x)\,\nabla (\zeta\cdot x)^{k-q}\cdot\nabla\big(|x|^{2-n-2q}(\zeta\cdot x)^q\big)\,\textup{d}\textup{V}(x) = 0. \]
By formula \eqref{eq:specialpair} this becomes
\[ \int_{\mathscr{A}_{a,b}} f(x)\, |x|^{-n-2q} (\zeta\cdot x)^k \,\textup{d}\textup{V}(x) = 0; \]
the scalar factor $-p(n+2q-2)$ has been omitted because it is nonzero.
Passing to the $n$-dimensional spherical coordinates $x=r\sigma$, $(r,\sigma)\in[a,b]\times\mathbb{S}^{n-1}$, we get
\begin{equation}\label{eq:momentnull}
\int_a^b \left(\int_{\mathbb{S}^{n-1}} f(r\sigma)(\zeta\cdot\sigma)^k\,\textup{d}\textup{S}(\sigma)\right)
 r^{k-1-2q}\,\textup{d}r = 0.
\end{equation}
By Lemma~\ref{lem:span}, the restrictions of $(\zeta\cdot\sigma)^k$ to $\mathbb{S}^{n-1}$ indexed by all such $\zeta$ span the whole $\mathscr{H}_k$, so by linearity \eqref{eq:momentnull} remains valid with $(\zeta\cdot\sigma)^k$ replaced by an arbitrary spherical harmonic $Y\in\mathscr{H}_k$.
In particular, it holds for $Y=\overline{Y_{k,j}}$, 
\[ \int_a^b \left(\int_{\mathbb{S}^{n-1}} f(r\sigma)\overline{Y_{k,j}(\sigma)}\,\textup{d}\textup{S}(\sigma)\right)
 r^{k-1-2q}\,\textup{d}r = 0. \]
Since formula \eqref{eq:sph_exp_coeff} applies for a.e.\@ $r\in[a,b]$, dividing by $\|Y_{k,j}\|_{\textup{L}^2(\mathbb{S}^{n-1})}^2$ we obtain
\[ \langle g_{k,j}(r), r^{k-n-2q} \rangle_{\textup{L}^2([a,b],r^{n-1}\textup{d}r)} = \int_a^b g_{k,j}(r)\,r^{k-n-2q}\,r^{n-1}\textup{d}r = 0. \]
Since this holds for every monomial $r^{k-n-2q}$ in $\mathfrak{R}_k$, we obtain $g_{k,j}\in\mathfrak{G}_k$.

\smallskip
\emph{Sufficiency.}
Assume now that $g_{k,j}\in\mathfrak{G}_k$ for every $k,j$ and that $f$ is given by the $\textup{L}^2(\mathscr{A}_{a,b})$-convergent series \eqref{eq:sph_exp}.
Let $u$ be a complex-valued harmonic function on $\mathscr{U}$, and let $v$ be a complex-valued harmonic function on $\mathscr{V}$ that vanishes at infinity.
The standard expansions in spherical harmonics \cite[Ch.~10]{ABR01} yield
\begin{align*}
u(x) & = \sum_{p=0}^{\infty} P_p(x); \quad x\in\mathscr{U}, \\ v(x) & = \sum_{q=q_0}^{\infty} (\mathcal{K}Q_q)(x); \quad x\in\mathscr{V},
\end{align*}
where $P_p\in\mathcal{H}_p$, $Q_q\in\mathcal{H}_q$, and $q_0=1$ for $n=2$, while $q_0=0$ for $n\geq3$.
Both series, together with their first derivatives, converge uniformly on the closed annulus $\mathscr{A}_{a,b}$.
Because of this and bilinearity, the verification of \eqref{eq:runcorn} boils down to
\begin{equation}\label{eq:runcorn_terms}
\int_{\mathscr{A}_{a,b}} f(x)\nabla P_p(x)\cdot\nabla (\mathcal{K}Q_q)(x)\,\textup{d}\textup{V}(x)=0
\end{equation}
for spherical harmonics $P_p\in\mathcal{H}_p$, $Q_q\in\mathcal{H}_q$ and integers $p\geq1$ and $q\geq q_0$.
In order to verify \eqref{eq:runcorn_terms} we apply the decomposition from Lemma~\ref{lem:decompose} and expand its left-hand side as
\begin{equation}\label{eq:runcorn_terms2}
\sum_{\ell} \int_{\mathscr{A}_{a,b}} f(x) \,|x|^{-n-2q+2\ell} H_{p+q-2\ell}(x) \,\textup{d}\textup{V}(x)
\end{equation}
for some solid spherical harmonics $H_{p+q-2\ell}\in\mathcal{H}_{p+q-2\ell}$.
\begin{itemize}
\item If $n\geq3$, then generalized spherical coordinates $x=r\sigma$ turn each term of \eqref{eq:runcorn_terms2} into
\[ \int_a^b \int_{\mathbb{S}^{n-1}} f(r\sigma) H_{p+q-2\ell}(\sigma) \,\textup{d}\textup{S}(\sigma) \,r^{-n+p-q} r^{n-1}\,\textup{d}r; \quad 0\leq\ell\leq \min\{p-1,q\}. \]
After substituting $k:=p+q-2\ell$, expanding $H_k$ into the orthogonal basis $(\overline{Y_{k,j}})_{1\leq j\leq d(k)}$, and recalling \eqref{eq:sph_exp_coeff}, it remains to check
\[ \int_a^b g_{k,j}(r) \,r^{-k-n+2(p-\ell)}r^{n-1}\,\textup{d}r = 0 \]
for
\[ k\geq 1,\ 1\leq j\leq d(k),\ p\geq1,\ \max\{0,p-k\}\leq\ell\leq p-1. \]
Indeed, all these integrals vanish because 
\[ r^{-k-n+2(p-\ell)} \in \mathfrak{R}_k = \{r^{-k-n+2}, r^{-k-n+4},\ldots, r^{k-n}\} \]
and we assumed $g_{k,j}\in\mathfrak{G}_k$.
\item If $n=2$, then the terms of \eqref{eq:runcorn_terms2} can be rewritten as
\[ \int_a^b \int_{\mathbb{S}^{1}} f(r\sigma) H_{p+q-2\ell}(\sigma) \,\textup{d}\textup{S}(\sigma) \,r^{-2+p-q} r\,\textup{d}r; \quad 0\leq\ell\leq \min\{p-1,q-1\}. \]
Substituting $k:=p+q-2\ell$ again this reduces to
\[ \int_a^b g_{k,j}(r) \,r^{-k-2+2(p-\ell)}r\,\textup{d}r = 0 \]
for
\[ k\geq 2,\ 1\leq j\leq d(k),\ p\geq1,\ \max\{0,p-k+1\}\leq\ell\leq p-1. \]
The last integral equality holds because now we know
\[ r^{-k-2+2(p-\ell)} \in \mathfrak{R}_k = \{r^{-k}, r^{-k+2},\ldots, r^{k-4}\} \]
and we took $g_{k,j}$ from the orthogonal complement $\mathfrak{G}_k$ of $\mathfrak{R}_k$.
\end{itemize}
This completes the proof of \eqref{eq:runcorn_terms} and then also that of \eqref{eq:runcorn}.
\end{proof}

\begin{proof}[Proof of Theorem~\ref{thm:varying}]
For fixed complex-valued $u$ and $v$ as in the theorem statement, define
\[ F(r):=\int_{\{x\in\R^n:\ |x|=r\}} f(x)\nabla u(x)\cdot\nabla v(x)\,\textup{d}\textup{S}(x); \quad r\in(r_-,r_+). \]
Equivalently written,
\[ F(r)=r^{n-1}\int_{\mathbb{S}^{n-1}} f(r\sigma)\nabla u(r\sigma)\cdot\nabla v(r\sigma)\,\textup{d}\textup{S}(\sigma), \]
so $F$ is continuous in $r$.
Passing to the generalized spherical coordinates then gives
\[ \int_{\mathscr{A}_{a,b}} f(x)\nabla u(x)\cdot\nabla v(x)\,\textup{d}\textup{V}(x)=\int_a^b F(r)\,\textup{d}r \]
for all $a<b$ from the interval $(r_-,r_+)$.
Hence, (\ref{thmit1}) and (\ref{thmit2}) are equivalent.

\smallskip
Assume next that (\ref{thmit1}) holds and define the functions $g_{k,j}$ precisely by the formula \eqref{eq:sph_exp_coeff} for every $r\in(r_-,r_+)$. Each coefficient function is clearly continuous in $r$. For every compact subinterval $[a,b]\subset(r_-,r_+)$ the $\textup{L}^2$-expansion holds for the restriction of $f$ to $\mathscr{A}_{a,b}$. Theorem~\ref{thm:fixed} applies and gives $g_{k,j}|_{[a,b]}\in\mathfrak{G}_k$ for $k\geq0$ and $1\leq j\leq d(k)$.

\begin{itemize}
\item If $n\geq3$, $k\geq1$, and $1\leq j\leq d(k)$, then $r^{k-n}\in\mathfrak{R}_k$ (obtained for $\ell=k-1$), so
\[ \int_a^b g_{k,j}(r) \,r^{k-1}\,\textup{d}r = \bigl\langle g_{k,j}, r^{k-n} \bigr\rangle_{\textup{L}^2([a,b],r^{n-1}\textup{d}r)} = 0 \]
for every $[a,b]\subset(r_-,r_+)$.
Hence the continuous function $r\mapsto g_{k,j}(r) \,r^{k-1}$ has integral $0$ over every subinterval of $(r_-,r_+)$, so it vanishes identically. Therefore $g_{k,j}=0$. From the expansion \eqref{eq:sph_exp} we conclude that $f$ is radial.
\item If $n=2$, $k\geq2$, and $1\leq j\leq d(k)$, then $r^{k-4}\in\mathfrak{R}_k$ (obtained for $\ell=k-2$), so
\[ \int_a^b g_{k,j}(r) \,r^{k-3}\,\textup{d}r = \bigl\langle g_{k,j}, r^{k-4} \bigr\rangle_{\textup{L}^2([a,b],r\textup{d}r)} = 0 \]
for every $[a,b]\subset(r_-,r_+)$.
Again, the continuous function $r\mapsto g_{k,j}(r) \,r^{k-3}$ has integral $0$ over every subinterval, so $g_{k,j}=0$.
Hence, only spherical harmonics of degrees $0$ and $1$ may occur and expansion \eqref{eq:sph_exp} shows that $f$ is of the form asserted in (\ref{thmit3}).
\end{itemize}
We have proved that (\ref{thmit1}) implies (\ref{thmit3}).

\smallskip
Conversely, assume (\ref{thmit3}) and fix any $a<b$ in $(r_-,r_+)$.
If $n\geq3$, only the zeroth degree coefficient of $f$ may be nonzero, and $\mathfrak{R}_0=\emptyset$, so Theorem~\ref{thm:fixed} shows that \eqref{eq:runcorn} holds on $\mathscr{A}_{a,b}$.
If $n=2$, only the coefficients of degrees $0$ and $1$ may be nonzero, while $\mathfrak{R}_0=\mathfrak{R}_1=\emptyset$, so Theorem~\ref{thm:fixed} again yields \eqref{eq:runcorn} on $\mathscr{A}_{a,b}$.
Since $a$ and $b$ were arbitrary, (\ref{thmit3}) implies (\ref{thmit1}).
\end{proof}

It is also possible to give a direct proof that (\ref{thmit2}) and (\ref{thmit3}) are equivalent, without relying on Theorem~\ref{thm:fixed}. Essentially, one only uses the general expansion \eqref{eq:sph_exp} and the basic properties of spherical harmonics; we omit the details.


\section{Examples}
\label{sec:examples}

We now write down explicit consequences of Theorem~\ref{thm:fixed}. Throughout this section the annulus $\mathscr{A}_{a,b}$ is fixed, with $r_-<a<b<r_+$.

\subsection{A general construction}
For every $k\geq0$, let
\[ \Pi_k\colon \textup{L}^2\big([a,b],r^{n-1}\,\textup{d}r\big)
\to \operatorname{span}\mathfrak{R}_k \]
be the orthogonal projection with respect to the inner product $\langle\cdot,\cdot\rangle_{\textup{L}^2([a,b],r^{n-1}\textup{d}r)}$.
Then, for any smooth function $\phi\in \textup{C}^{\infty}([a,b])$ and any nonzero spherical harmonic
$Y\in\mathscr{H}_k$, the function
\[ f_{k,\phi,Y}(r\sigma):=\big(\phi(r)-\Pi_k\phi(r)\big)Y(\sigma) \]
is a smooth solution of \eqref{eq:runcorn} on the fixed annulus $\mathscr{A}_{a,b}$.
Indeed, the radial factor $\phi-\Pi_k\phi$ belongs to $\mathfrak{G}_k$ by construction, and all other spherical harmonic coefficients vanish.
Thus, Theorem~\ref{thm:fixed} applies directly.
If $k\geq1$ and $\phi\not\in\operatorname{span}\mathfrak{R}_k$, this gives a nonzero non-radial solution.

The first few low-degree moment conditions are as follows. In Table~\ref{tab:low-degree}, the function $g$ denotes the radial coefficient of a fixed spherical harmonic $Y_{k,j}$ of degree $k$; the index $j$ is suppressed.

\begin{table}[ht]
\[ \begin{array}{|c||c|c|}
\hline
& n=2 & n=3 \\
\hline\hline
k=0 \rule{0mm}{5mm} & \text{no condition} & \text{no condition} \\[1mm]
\hline
k=1 \rule{0mm}{8mm} & \text{no condition} & \displaystyle \int_a^b g(r)\,\textup{d}r=0 \\[4mm]
\hline
k=2 \rule{0mm}{8mm} & \displaystyle \int_a^b g(r) \,r^{-1}\,\textup{d}r=0 & \displaystyle \int_a^b g(r) \,r^{-1}\,\textup{d}r=0,\ \int_a^b g(r) \,r\,\textup{d}r=0 \\[4mm]
\hline
\end{array} \]
\caption{Low-dimensional moment conditions.}
\label{tab:low-degree}
\end{table}
 
This table is just the low-degree specialization of the general formulas in Theorem~\ref{thm:fixed}; higher degrees add one further moment condition for each additional admissible Laurent monomial.

\subsection{Two-dimensional examples}
In dimension $n=2$, a convenient basis of $\mathscr{H}_k$ on the unit circle is provided by the complex exponentials:
\begin{align*} 
& Y_{0,1}(e^{\mathbbm{i}\varphi}) := 1, \\
& Y_{k,1}(e^{\mathbbm{i}\varphi}) := e^{\mathbbm{i}k\varphi},\quad
Y_{k,2}(e^{\mathbbm{i}\varphi}) := e^{-\mathbbm{i}k\varphi}; \quad k\geq1. 
\end{align*}
Hence every $f\in \textup{L}^2(\mathscr{A}_{a,b})$ can be written as
\[ f(r,\varphi)=g_0(r) + \sum_{k=1}^{\infty} g_k^+(r)e^{\mathbbm{i}k\varphi} + \sum_{k=1}^{\infty} g_k^-(r)e^{-\mathbbm{i}k\varphi}. \]
Theorem~\ref{thm:fixed} says that \eqref{eq:runcorn} holds on a fixed annulus if and only if $g_0$, $g_1^+$, and $g_1^-$ are arbitrary, while 
\[ \int_a^b g_k^{\pm}(r) \,r^{2\ell-k+1}\,\textup{d}r = 0;
\quad \ell=0,1,\ldots,k-2 \]
for every $k\geq2$.
Thus the parts of degree $1$ are unconstrained, whereas every higher angular frequency is subject only to finitely many moment conditions.

For instance,
\[ f_1(r,\varphi):=r\cos\varphi \]
is not radial and satisfies \eqref{eq:runcorn} on every annulus.
On the other hand,
\[ f_2(r,\varphi):=\left(r-\frac{b-a}{\log(b/a)}\right)\cos(2\varphi) \]
satisfies \eqref{eq:runcorn} on the fixed annulus $\mathscr{A}_{a,b}$, because the only condition for degree $k=2$ stated in Table~\ref{tab:low-degree} is
\[ \int_a^b g(r) \,r^{-1}\,\textup{d}r=0, \]
and the chosen radial factor has exactly this property.

\subsection{Three-dimensional examples}
In dimension $n=3$, let 
\[ \{Y_k^m \,:\, -k\leq m\leq k\} \]
be the usual basis of spherical harmonics of degree $k$ on $\mathbb{S}^2$.
In spherical coordinates
\[ \begin{pmatrix} x\\ y\\ z\end{pmatrix} 
= \begin{pmatrix} r\sin\theta\cos\varphi\\ r\sin\theta\sin\varphi\\ r\cos\theta \end{pmatrix}; \quad r\in[0,\infty),\ \theta\in[0,\pi],\ \varphi\in[0,2\pi), \]
every $f\in \textup{L}^2(\mathscr{A}_{a,b})$ has an expansion
\[ f(r,\theta,\varphi)=\sum_{k=0}^{\infty}\sum_{m=-k}^{k} g_{k,m}(r)Y_k^m(\theta,\varphi), \]
and the conditions from Theorem~\ref{thm:fixed} become
\[ \int_a^b g_{k,m}(r) \,r^{-k+1+2\ell}\,\textup{d}r = 0;
\quad \ell=0,1,\ldots,k-1 \]
for every $k\geq1$ and every $m\in\{-k,\ldots,k\}$.
In particular, degree $1$ requires only one condition, namely
\[ \int_a^b g_{1,m}(r)\,\textup{d}r=0, \]
which is in line with Table~\ref{tab:low-degree} again.
Therefore,
\[ f(r,\theta,\varphi):=\Bigl(r-\frac{a+b}{2}\Bigr)\cos\theta \]
satisfies \eqref{eq:runcorn} on the fixed annulus $\mathscr{A}_{a,b}$.
This gives an explicit non-radial solution in three dimensions.
By contrast, Theorem~\ref{thm:varying} shows that no nonzero degree-$1$ term can survive if the identity is required on every sub-annulus.
Adding a radial term gives strictly positive non-radial examples on the same fixed annulus.
Namely, for $0<\varepsilon<2/(b-a)$ set
\[ f_{\varepsilon}(r,\theta,\varphi):=1+\varepsilon\Bigl(r-\frac{a+b}{2}\Bigr)\cos\theta \]
and note that $f_{\varepsilon}$ is strictly positive in addition to satisfying \eqref{eq:runcorn} on $\mathscr{A}_{a,b}$.


\section*{Author contributions}
Both authors contributed equally to this work.


\section*{Declaration of AI usage}
During the preparation of this manuscript, OpenAI's ChatGPT 5.4 Pro and ChatGPT 5.5 Pro were used for literature search, computation verification, text proofreading, exposition improvement, and production of TikZ code of Figure~\ref{fig:annulus-geometry}. The ideas, the proofs, the bibliography, and the writing of the manuscript are entirely the work of the authors.


\section*{Acknowledgments}
V.\,K. was supported by the Croatian Science Foundation project no.~IP-2022-10-5116, \emph{Fourier analysis and applications}, and by the European Union -- NextGenerationEU through the National Recovery and Resilience Plan 2021--2026, via an institutional grant of the University of Zagreb Faculty of Science IK IA 1.1.3. Impact4Math. 

I.\,S. was supported by the Croatian Science Foundation project no.~IP-2025-02-8625, \emph{Quantum aspects of gravity}, and by the European Union -- NextGenerationEU through the National Recovery and Resilience Plan 2021--2026, via an institutional grant of the University of Zagreb Faculty of Science (ProPuBFO).


\end{document}